\documentclass[11pt]{amsart}
\usepackage{amsmath}
\usepackage{amsfonts}
\usepackage{amssymb}
\usepackage{bbm}
\usepackage{fullpage}
\usepackage{verbatim}
\usepackage{color}
\usepackage{datetime}
\usepackage{cite}

\numberwithin{equation}{section}

\newcommand{\rr}{\m{R}}

\newcommand{\nn}{\m{N}}
\newcommand{\cc}{\m{C}}
\newcommand{\zz}{\m{Z}}

\newcommand{\1}{\mathbbm{1}} 
\newcommand{\spec}[1]{\mathrm{sp}(#1)}

\newcommand{\la}{\lambda}
\newcommand{\La}{\Lambda}
\newcommand{\sig}{\sigma}
\newcommand{\ep}{\varepsilon}
\newcommand{\del}{\delta}

\newcommand{\set}[1]{\left\{ #1 \right\} }
\newcommand{\ip}[1]{\langle #1 \rangle}
\newcommand{\norm}[1]{\left \Vert #1 \right\Vert }

\newcommand{\m}[1]{\mathbb{#1}}
\newcommand{\mc}[1]{\mathcal{#1}}
\newcommand{\mf}[1]{\mathfrak{#1}}

\newtheorem{theorem}{Theorem}[section]
\newtheorem{lemma}[theorem]{Lemma}
\newtheorem{proposition}[theorem]{Proposition}
\newtheorem{corollary}[theorem]{Corollary}
\newtheorem{definition}[theorem]{Definition}
\newtheorem{example}[theorem]{Example}

\newcommand{\specc}{\mathrm{sp}}

\newcommand{\diam}[1]{\text{diam}(#1)}

\newcommand{\supp}[1]{\text{supp}(#1)}

\newcommand{\Aut}{\mathrm{Aut}}

\newcommand{\an}{\mathrm{An}}
\newcommand{\crr}{\mathrm{Cr}}

\newcommand{\A}{\mathfrak{A}}

\title{Automorphic equivalence preserves the split property} 
\author{Alvin Moon}
\address{Department of Mathematics. University of California, Davis. Davis, CA. 95616. USA}
\email{asmoon@math.ucdavis.edu}

\begin{document}

\maketitle

\date{\today}
\begin{abstract}
We prove that the split property is a stable feature for spin chain states which are related by composition with $*$-automorphisms generated by power-law decaying interactions. We apply this to the theory of the $\mathbb{Z}_2$-index for gapped ground states of symmetry protected topological phases to show that the $\mathbb{Z}_2$-index is an invariant of gapped classification of phases containing fast-decaying interactions.    
\end{abstract}

\section{Introduction}

Recent studies have rigorously proven the existence of symmetry protected topological (SPT) phase transitions in one dimension using an invariant of smooth, gapped classification known as the $\zz_2$-index \cite{O, T}. The $\zz_2$-index follows a line of investigation of the invariants which arise from symmetries of a quantum spin chain \cite{BN, NT, PTBO, PTBO2}. 

\medskip 

In \cite{O}, it is proven that this index is a well-defined invariant for finite-range interactions, regardless of boundary conditions, and that the index agrees with the matrix product state index defined in \cite{PTBO}. Thus it is concluded that the AKLT interaction belongs to a non-trivial topological phase of finite-range interactions protected by time reversal symmetry. It is also known that extensive but sufficiently small and fast-decaying perturbations of the AKLT interaction on the chain will not move the system out of the phase (e.g. see \cite{MN}). In particular, an SPT phase may contain interactions which are not finite range. 

\medskip 

The objective of this note is to investigate when the $\zz_2$-index is a stable invariant of an SPT phase in one dimension. We prove that under certain hypotheses, including superpolynomial but still subexponential decay of interactions and uniqueness of the gapped ground state, that if an SPT phase contains an interaction with a well-defined $\zz_2$-index, then all interactions in the phase have a well-defined index, and that the index is an invariant of the phase (see Section 3 for the hypotheses). For this, we follow the proof of Ogata for the finite-range case closely, making the necessary and material modifications to handle an unbounded range of interaction. This stability provides further evidence that the $\zz_2$-index detects a true phase transition between interactions in distinct symmetry protected topological phases. 

\medskip 

A significant mathematical obstruction to assuming weaker decay conditions is in proving that certain gapped ground states of interactions satisfy the \textit{split property}. So far, general results on sufficient conditions for the split property to hold critically use characteristics of finite-range or exponentially decaying one-dimensional interactions, such as boundedness of the entanglement entropy or the validity of Haag duality for the spin chain interactions \cite{M, M3}. We comment on the relationship between split property for translation invariant ground states and Haag duality in Section 2. 

\medskip 

Our main result is that quasi-local deformations of split states preserve the split property. Our proofs make use of Lieb-Robinson bounds on the speed of propagation of time-evolved observables which do not depend on the sizes of support. To the best of our knowledge, the results of this note are the first which generally guarantee the split property for ground states of interactions which do not necessarily decay exponentially.

\subsection{Notations and assumptions}

\medskip

We consider the one-dimensional lattice $(\zz, |\cdot|)$. Let $P_f (\Sigma)$ denote the finite subsets of $\Sigma \subset \zz$. The onsite Hilbert space at $x\in \zz$ is $\mf{H}_x = \cc^d$, where $d\geq 2$ is taken to be independent of $x$ for simplicity. Let $\A_{\set{x}} = M_d(\cc)$ denote the onsite algebra of observables. Local algebras of observables for $X \in P_f(\zz)$ are defined by tensor product:
	\begin{equation}
		\begin{split}
\A_X = \bigotimes _{ x\in X} \A_{\set{x}}.
		\end{split}
	\end{equation} 
We reserve $\La$ as notation for a finite interval of the form $[a,b] \cap \zz$. Let $\A_{loc}$ denote the maximal algebra obtained by inclusion of local algebras, and $\A_\zz$ its closure with respect to the operator norm:
	\begin{equation}
		\begin{split}
\A_{loc} = \bigcup_{ X \in P_f(\zz)} \A_X, \hspace{10mm} \A_\zz = \overline{ \A_{loc}} ^{ \norm{\cdot}}.
		\end{split}
	\end{equation}
Similarly, let $\A_L$ and $\A_R$ denote the $C^*$-algebras obtained from the local algebras of the left and right complementary half-infinite chains, respectively: 
	\begin{equation}
		\begin{split}
\A _L = \overline{ \bigcup _{X \in P_f( (-\infty, 0] ) }\A_X } ^{\norm{\cdot}} \hspace{5mm} (\text{resp. } \A_R).  
		\end{split}
	\end{equation}
We model the interactions between sites of the lattice with \textit{interaction functions} parametrized by a dependence $t \in [0,1]$:
	\begin{gather}
 \Phi( \cdot, t): P_f(\zz)  \to \A _\zz  \\
X  \mapsto \Phi(X,t) = \Phi(X,t)^* \in \A _X .
	\end{gather} 
For regularity, we assume for each $X\in P_f(\zz)$ that the dependence $t \mapsto \Phi(X, t)$ is continuously differentiable. The dynamics $\tau^\La : [0,1] \to \mathrm{Aut}( \A_\La)$ of the model are generated by the family of Hamiltonians:
	\begin{equation}\label{eq:Hamiltonian}
		\begin{split}
H_\La (\Phi,t ) = \sum_{Z \subset \La} \Phi(Z, t),
		\end{split}
	\end{equation} 	 
are continuous in $t$ and satisfy $\tau_0^\La = \mathrm{id}$. For a thorough investigation of properties of $\Phi(\cdot, t)$ and the limit of the family $(\tau^\La)$, the curve  $\tau: [0,1] \to \Aut(\A_\zz)$ of $*$-automorphisms, we refer to Section 3 of \cite{NSY}. In this case, we say $\tau$ are the \textit{quasi-local dynamics generated by} $\Phi(\cdot,t)$.

\medskip 

In this note, we study antilinear symmetries of the spin chain. Precisely,  if $\theta: \A_\zz \to \A_\zz$ is an antilinear $*$-automorphism, we say that $\tau$ is $\theta$-invariant if the generating interactions are fixed by $\theta$:
	\begin{equation}
		\begin{split}
\forall X \in P_f(\zz): \theta\big{(} \Phi(X, t) \big{)} = \Phi(X,t).
		\end{split}
	\end{equation}
Physical considerations require decay of the interaction. To account for the $t$-dependence, we quantify the decay using $\mc{F}$-functions and $\mc{F}$-norms. Precisely, denote:
	\begin{equation}\label{eq:Ffunction}
		\begin{split}
F_\beta(x) = e^{ - h(x) } \frac{1}{ (1+x)^\beta}, ~ \beta>0.
		\end{split}
	\end{equation}	  
where $h: [0,\infty) \to [0,\infty)$ is a non-negative, non-decreasing, subadditive function. We observe that there exists a constant $C_\beta$ such that for any $x,y\in \zz$: $\sum _{z\in \zz} F_\beta( |x-z|) F_\beta(|z-y|) \leq C_\beta F_\beta( |x-y|)$. We refer to $C_\beta$ as the \textit{convolution constant} of $F_\beta$.

\medskip 

The function $F_\beta$ depends on $h$, but we will suppress the $h$-dependence in notation, and when the choice of $\beta>0$ is immaterial, we will suppress the $\beta$-dependence as $F = F_\beta$. Then $\norm{\cdot}_{h,\beta}$ is defined for the family $\Phi(\cdot, t)$ as:
	\begin{equation}
		\begin{split}
\norm{\Phi}_{h,\beta} = \sup _{x,y\in \zz} \sum _{ \substack{ Z \in P_f(\zz)\\ x,y\in Z}} \sup _{t \in [0,1]} \bigg{(} \frac{\norm{\Phi(Z,t)}}{F_\beta(|x-y|)} \bigg{)}.
		\end{split}
	\end{equation} 	
We note that if $\norm{\Phi}_{h,\beta}  <\infty$, then the interaction decays uniformly as a function of the diameter:
	\begin{equation}
		\begin{split}
\norm{ \Phi(X,t)} \leq \norm{\Phi}_{h,\beta}  F_\beta( \diam{X}).
		\end{split}
	\end{equation}
We will denote:
	\begin{equation}
		\begin{split}
\norm{\Phi(Z)}_{[0,1]} = \sup _{t\in [0,1]} \norm{\Phi(Z,t) }.
		\end{split}
	\end{equation}	
Lastly, we state the split property from \cite{M,M3} which will be best suited for our analysis.  
\begin{definition}\label{def:split}
A state $\omega$ of $\A_\zz$ \textbf{satisfies the split property} if there exist states $\omega_L$ and $\omega_R$ of the left and right algebras $\A_{L}, \A _{R}$, respectively, such that $\omega$ is quasi-equivalent to $\omega_L \otimes \omega_R$. 
\end{definition}
For brevity, we will refer to states which satisfy Definition \ref{def:split} as \textit{split} states. The formulation of the split property in Definition \ref{def:split} agrees with that of e.g. \cite{M2,O} when $\omega$ is pure and $\omega_L$ and $\omega_R$ are the restrictions of $\omega$ to the left and right algebras, respectively. There are higher-dimensional generalizations of the split property, such as the distal or approximate split property of \cite{N}; however, we do not comment on whether these are stable. 

\medskip 

We will express the quasi-equivalence relation between states by $\sim$. We consider only factor states, so we recall an asymptotic condition for quasi-equivalence of factor states $\omega$ and  $\varphi$ of the quasi-local algebra $\A_\zz$ (cf. Corollary 2.6.11 in \cite{BR1}): $\omega \sim \varphi$ if and only if for all $\ep>0$, there exists $X_\ep \in P_f(\zz)$ such that $Y\in P_f(\zz)$ and $B \in \A_Y$ with $Y \cap X_\ep = \emptyset$ imply: 
	\begin{equation}
		\begin{split}
| \omega (B) - \varphi(B)| \leq \norm{B} \ep .
		\end{split}
	\end{equation}

\section{Split states}

\subsection{Support-independent Lieb-Robinson bounds }

In the following, we prove special cases of Lieb-Robinson bounds for the integer lattice and certain configurations of supports. These bounds will be useful in proving Theorem \ref{theorem:stability-of-split}.

\medskip 

Let $n,m\in \nn_+$ such that $m<n$. For ease of notation, we define the following family of sets: 
	\begin{equation}\label{eq:ann-cr}
		\begin{split}
\an(m,n) & = [-n, -m] \cup [m, n] 
		\end{split}
	\end{equation}
\begin{lemma}\label{lem:LR2}
Suppose $\norm{\Phi}_{h,\beta} <\infty$ for $\beta > 2$. There exists a constant $\kappa(\beta )>0$ such that for all choices $n,m,c,p\in \nn$ with $c<m<n$, if $[ -n-p, n+p ] \subset \La$, the following inequality holds for all $A \in \an(m,n)$ and $B \in \La \setminus \an( m-c, n+p)$:
	\begin{equation}\label{eq:LR-conc}
		\begin{split}
\norm{ [\tau^{\La}_t( A), B] } \leq \kappa(\beta) \norm{A} \norm{B} ( e^{\nu |t|} - 1)  F_{\beta-2} ( \min \set{ p,c} )
		\end{split}
	\end{equation} 
where $\kappa(\beta)$ and $\nu$ can be taken as:
	\begin{equation}
		\begin{split}
\kappa(\beta) & = \frac{ 16  }{C_\beta}(\beta/2-1)^{-2}
\\ 
\\
\nu & = 2 \norm{\Phi}_{h,\beta}  C_\beta. 
		\end{split}
	\end{equation} 	
In particular, $\kappa(\beta)$ does not depend on $n,m,c, p$ or the function $h(x)$. 
\end{lemma}

\begin{proof}
Denote $\crr(a,b) = \Lambda \setminus \an(a,b)$. By iterative arguments (cf. \cite{NS2}), it can be shown that the commutator in (\ref{eq:LR-conc}) is bounded above by the series:
	\begin{equation}
		\begin{split}
\sup _{\substack{A \in \A_{\an(m,n)} \setminus \set{0} \\ \norm{A}=1}} \norm{ [ \tau_t^\La(A), B]} \leq 2 \norm{B} \sum _{k=1}^\infty \frac{ (2|t|)^k}{k!} a_k 
		\end{split}
	\end{equation}
where the right-hand side is convergent for $a_k$ defined:
	\begin{gather}
a_k = \sum _{Z_1 \in S_\La\big{(} \an(m,n) \big{)}} \sum _{Z_2 \in S_\La( Z_1)}  \cdots \sum _{Z_k \in S_\La(Z_{k-1})} \delta_Y (Z_k) \norm{\Phi(Z_1)}_{[0,1]} \cdots \norm{\Phi(Z_k)}_{[0,1]} \\
\del_Y(W) = \bigg{ \{ }\begin{array}{l l} 1 & \text{ if }W \cap Y \not = \emptyset \\ 0 & \text{ else} \end{array}.
	\end{gather}
Here, $S_\La(W) = \set{ Z \subset \La : Z \cap W\not = \emptyset, Z \cap W^c \not = \emptyset}$ denotes the boundary sets of $W$. Let $C_\beta$ denote the convolution constant of $F_\beta$. For any $k$:
	\begin{equation}\label{eq:LR-work}
		\begin{split}
a_k & \leq \sum _{\substack{x\in \an(m,n)\\ y\in \crr(m-c,n+p)}} ~ \sum _{z_1, \ldots, z_{k-1} \in \La} ~\sum _{\substack{ Z_ 1 \in S_\La (\an(n,m) ) \\ x, z_1 \in Z_1}} \sum _{\substack{ Z_ 2 \in S_\La(Z_1) \\ z_1, z_2 \in Z_2}} \cdots \sum _{\substack{ Z_ {k} \in S_\La(Z_{k-1}) \\ z_{k-1}, y \in Z_k}} \norm{\Phi(Z_1)}_{[0,1]} \cdots \norm{\Phi(Z_k)}_{[0,1]} \\
& \leq \sum _{\substack{x\in \an(m,n) \\y\in \crr(m-c,n+p)}} \sum _{z_1 \in \La} ~ \sum _{\substack{ Z_ 1 \in S_\La(\an(m,n)) \\ x, z_1 \in Z_1}} \norm{\Phi(Z_1)}_{[0,1]} ( C_\beta^{k-2} \norm{\Phi}_{h,\beta}^{k-1} ) F_\beta(|z_1 - y|) \\
& \leq ( \norm{\Phi}_{h,\beta} ^{k} C_\beta ^{k-1}) e^{ - h( \min\set{p,c})}\sum _{x \in \an(m,n)} ~ \sum _{y\in \crr (m-c,n+p)} \frac{1}{ (1+|x-y|)^\beta} .
		\end{split}
	\end{equation} 
Since:
	\begin{equation}
		\begin{split}
\sum _{m \leq x \leq n} ~ \sum _{y \in \crr(m-c,n+p)}  \frac{1}{ (1+|x-y|)^\beta}& \leq 2\sum _{m\leq x \leq n}  \frac{1}{ (1+ d( x, \crr(m-c,n+p) )^{\beta/2}}  \sum _{r \geq \min \set{p,c}} \frac{1}{(1+r)^{\beta/2}} \\
& \leq 4 \bigg{(} \sum _{r \geq \min \set{p,c}} \frac{1}{(1+r)^{\beta/2}} \bigg{)}^2 
		\end{split}
	\end{equation}
the symmetry in the sum of the last inequality of (\ref{eq:LR-work}) implies:
	\begin{equation}
		\begin{split}
a_k & \leq 8 \norm{\Phi}_{h,\beta}^k C_\beta^{k-1} e^{ - h( \min\set{p,c})}\bigg{(} \sum _{r \geq \min \set{p,c}} \frac{1}{(1+r)^{\beta/2}} \bigg{)}^2 \\
& \leq 8 \norm{\Phi}^k_{h,\beta} C_\beta^{k-1} ( \beta/2 - 1)^{-2} e^{ - h( \min\set{p,c})} \frac{1}{ ( 1+\min \set{p,c})^{ \beta - 2}} \\
\\
& \leq 8 \norm{\Phi}^k_{h,\beta} C_\beta^{k-1} ( \beta/2 - 1)^{-2} F_{\beta-2}( \min\set{p,c}).
		\end{split}
	\end{equation} 
Hence the inequality (\ref{eq:LR-conc}) holds with the choices: 
	\begin{equation}
		\begin{split}
\kappa(\beta ) & = \frac{ 16 }{C_\beta }(\beta/2-1)^{-2}  \\
\nu & = 2 \norm{\Phi} _{h, \beta} C_\beta.
		\end{split}
	\end{equation}
\end{proof}

We also record for completeness the following useful bound. 

\begin{corollary}\label{lem:LR}
Suppose $\norm{\Phi}_{h,\beta} < \infty$ for $\beta>2$. If $X, Y \subset \La$ with $\max X < \min Y$, then for all $A \in \A_X$, $B \in \A_Y$, 
	\begin{equation}\label{eq:bound}
		\begin{split}
\norm{ [ \tau_t^\La (A), B] } \leq \kappa(\beta) \norm{A} \norm{B} ( e^{\nu |t|} - 1) F_{\beta-2}\big{(} d(X,Y) \big{)}.
		\end{split}
	\end{equation}
\end{corollary}

\begin{proof}
The conclusion follows from observing that the origin has no distinguished role in the proof of Lemma \ref{lem:LR2}. 
\end{proof}

We remark that taking the $\La \to \zz$ limit in Lemma \ref{lem:LR2} and Corollary \ref{lem:LR} shows that the infinite volume dynamics $\tau$ also satisfies the corresponding support-independent Lieb-Robinson bound.

\subsection{Automorphic equivalence and the split property} 

\medskip

We say states $\omega$ and $\varphi$ of $\A_\zz$ are \textit{automorphically equivalent} if there exist quasi-local dynamics $\tau: [0,1]\to \Aut(\A_\zz)$ such that:
	\begin{equation}
		\begin{split}
\omega = \varphi \circ \tau_1 . 
		\end{split}
	\end{equation}
In this section we prove that the split property is stable under automorphic equivalence. To proceed, we remark that if $\omega$ is a split factor state,  $\omega \sim \omega_L \otimes \omega_R$, then for any $\beta\in \Aut( \mathfrak{A}_\zz)$, the following states are also factor: $\omega_L \otimes \omega_R$, $\omega_L \otimes \omega_R \circ \beta$ and $\omega \circ \beta$. Next, let $\Phi^{L}(\cdot, t) : P_f( ( -\infty, 0]) \to \bigcup _{Z \subset (-\infty, 0]} \A_Z$ denote the restriction of $\Phi(\cdot,t)$ to the left half-infinite chain. Define $\Phi^R(\cdot,t)$ the same way using the complementary right half-infinite chain. $\Phi^L(\cdot,t)$ generates quasi-local dynamics $\tau^L :[0,1]  \to \Aut( \A _L)$ (resp. $\tau^R$). Likewise, the interaction $\Phi^{\cup}(\cdot,t): P_f(\zz) \to \A_{loc}$ defined by:
	\begin{equation}
		\begin{split}
\Phi^\cup (X,t) = \bigg{ \{ } \begin{array}{l l} \Phi(X,t) & \text{ if } X\subset (-\infty, 0] \text{ or } X \subset [1, \infty) \\ 0 & \text{ else } \end{array}
		\end{split}
	\end{equation}
generates quasi-local dynamics $\tau^\cup:[0,1] \to \Aut(\A_\zz).$ Then in the notation:
	\begin{equation}
		\begin{split}
( \omega_L\circ \tau_t^L) \otimes (\omega_R \circ \tau^R_t) = (\omega_L \otimes \omega_R) \circ \tau_t^\cup .
		\end{split}
	\end{equation}
	
In the following theorem, we consider interactions which decay by at least a power law, setting $h$ in (\ref{eq:Ffunction}) to be the zero function. 

\begin{theorem}\label{theorem:stability-of-split}
Suppose $\tau:[0,1] \to \Aut(\A_\zz)$ are quasi-local dynamics with a generating interaction $\Phi(\cdot, t)$ such that $\norm{\Phi}_{0,\beta} < \infty$. If $\omega_0$ is a split factor state and $\beta>3$, then $\omega_t = \omega_0 \circ \tau_t$ is also a split factor state, for all $t\in [0,1]$. 
\end{theorem}

\begin{proof}

Denote by $\omega_{L,t} =\omega_L \circ \tau_t^L$ (resp. $\omega_{R,t}$) and $\omega _0  = \omega$. Suppose $\ep>0$ and $n, r \in \nn$ such that $r>n$. Recalling the sets $\an(a,b)$ in (\ref{eq:ann-cr}), let $\m{E}_{n,r}: \A_\zz \to \A _{ \an(n, 2(n+r))}$ denote the conditional expectation with respect to the product trace state. Since $\omega$ is split and factor, there exists $N_\omega(\ep) \in \nn$ such that $ n > N_\omega(\ep)$ implies: 
	\begin{equation}
		\begin{split}
| \omega \circ \m{E} _{n,r} ( \tau_t(A)) - \omega_L \otimes \omega_R \circ \m{E}_{n,r}( \tau_t ( A)) | \leq \ep \norm{ \mathbb{E}_{n,r}( \tau_t(A))} \leq \ep \norm{A}. 
		\end{split}
	\end{equation} 
The following bounds will be derived independently of $r$, and so we will be able to let $r$ tend to infinity. Evidently for $A \in \A_{loc}$:
	\begin{equation}
		\begin{split}
| \omega_t (A) - \omega_{L,t}\otimes \omega_{R,t} (A) | & \leq  | (\omega - \omega_L \otimes \omega_R) \circ \tau_t (A) | + |\omega_L \otimes \omega_R \big{(} \tau_t (A) - \tau_t^\cup (A) \big{)} |  \\
& \leq \bigg{(} | \omega \circ \m{E} _{n,r} ( \tau_t(A)) - \omega_L \otimes \omega_R \circ \m{E}_{n,r}( \tau_t ( A)) | + 2 \norm{ \tau_t (A) - \m{E}_{n,r}(\tau_t(A)) }\bigg{)} \\
& \hspace{20mm} + \norm{ \tau_t(A) - \tau^\cup _t (A)}.
		\end{split}
	\end{equation}
Lemma \ref{lem:LR2} implies that if $\supp{A} \subset \an( 2n, 2n+r)$: 
	\begin{equation}
		\begin{split}
\norm{ \tau_t (A) - \m{E} _{n,r} ( \tau_t(A)) } \leq 2 \kappa(\beta) \norm{A} (e^{\nu |t| }-1) F_{\beta-2}(n). 
		\end{split}
	\end{equation}
To conclude the proof, it is left to show that for fixed $t\in \rr$, the quantity $\norm{ \tau_t(A) - \tau^\cup _t (A)}$ decays as a function of $n$, uniformly in the norm of $A$. This will follow from a Gronwall-type inequality. Let $\La$ be any interval containing $[ -2(n+r), 2(n+r)]$. Define:
	\begin{equation}
		\begin{split}
f_\La (t) = \tau_t^\La (A) - \tau_t^{\cup, \La} (A) 
		\end{split}
	\end{equation}
where $\tau^\La$ and $\tau^{\cup, \La}$ are the corresponding finite-volume dynamics. Since $f_\La(t)$ satisfies the ODE and initial value problem:
	\begin{equation}
		\begin{split}
\frac{d}{dt} f_\La(t) & = i [ H _\La( \Phi^\cup, t) , f_\La(t)] + i [ H_\La (\Phi,t) - H_\La ( \Phi^\cup, t) , \tau_t^\La (A) ]  \\
f_\La(0) & = 0
		\end{split}
	\end{equation}
the following bound is valid:
	\begin{equation}\label{eq:up-bound}
		\begin{split}
\norm{f_\La(t)} & \leq \norm{f _\La(0)} + \int _ 0 ^ {|t|} ds ~ \norm{ [ H_\La (\Phi,s) - H_\La ( \Phi^\cup, s) , \tau_s^\La(A) ]  } \\
& = \int _ 0^{|t|}ds~ \bigg{ \lVert}  \sum _{\substack{ Z \subset \La: \\Z\cap (-\infty, 0]\not = \emptyset \\ Z\cap [1,\infty) \not =\emptyset } } [ \Phi(Z,s), \tau_s ^\La (A)]  \bigg{ \rVert}.
		\end{split}
	\end{equation}
We can further divide the admissible $Z$ in the sum of the last line of  (\ref{eq:up-bound}) into:
	\begin{equation}
		\begin{split}
C_I & = \set{ Z \subset \Lambda : Z \cap (-\infty,0] \not = \emptyset, Z \cap [1,\infty) \not = \emptyset, Z \subseteq [-n,n]} \\ 
C_{II} & = \set{ Z\subset \Lambda : Z \cap (-\infty,0] \not = \emptyset, Z \cap [1,\infty) \not = \emptyset, Z \not \subseteq [-n,n]  } .
		\end{split}
	\end{equation}
The contribution of the $C_{II}$ terms to the upper bound in (\ref{eq:up-bound}) are majorized using decay of the interaction. Let $\delta>0$ such that $\beta > 2 + \delta$. Then:
	\begin{equation}\label{eq:1}
		\begin{split}
\norm{ \sum _{Z \in C_{II} } [ \Phi(Z,s), \tau_s^\La(A)] } & \leq \sum _{\substack{ x\in (-\infty, -n] \\ y\in [1, \infty)}}  2\norm{A} \bigg{(} \sum \set{ \norm{\Phi(Z,s)}: x,y\in Z} \bigg{)} \\
& \hspace{20mm} + \sum _{\substack{x\in [n,\infty) \\ y\in (-\infty, 0] }}  2\norm{A} \bigg{(} \sum \set{ \norm{\Phi(Z,s)}: x,y\in Z} \bigg{)} \\
& \leq 4 \norm{A}\norm{\Phi}_{0,\beta} \sum _{x=n}^\infty \sum _{y=0}^\infty F_\beta(x+y) \\
& \leq 4 \norm{A}\norm{\Phi}_{0,\beta} \sum _{x=n}^\infty F_{\beta-1-\delta/2}( x+1) \sum _{y=1}^\infty \frac{1}{(1+x+y)^{1+\delta/2}} \\
& \leq \bigg{[ } \frac{8}{(\beta-2-\delta/2)\delta} \norm{A} \norm{\Phi}_{0,\beta}   \bigg{]}  F_{\beta - 2-\delta/2} (n).
		\end{split}
	\end{equation}
And an application of Lemma \ref{lem:LR2} majorizes the contribution from $C_{I}$. Note we have the simple bound $\norm{\sum _{Z \in C_I} \Phi(Z,s) } \leq 3 \norm{\Phi}_{0, \beta} n $. And so: 
	\begin{equation}\label{eq:2}
		\begin{split}
\norm{ \sum _{Z \in C_I} \big{[} \Phi(Z,s), \tau_s^\La(A)\big{]} } &   \leq [ 3\kappa(\beta) \norm{\Phi}_{0,\beta}  \norm{A} ] (e ^{ \nu |s|} - 1) n F_{\beta-2}(n)   .
		\end{split}
	\end{equation}
The conclusion follows from the fact that the upper bounds in (\ref{eq:1}) and (\ref{eq:2}) are independent of the sufficiently large, finite interval $\La$ and $r$. 
\end{proof}

Lastly, we remark on when the left and right states in Theorem \ref{theorem:stability-of-split} can be taken to be restrictions. 

\begin{corollary}\label{cor:restrictions}
Suppose $\omega_0$ is a factor state such that $\omega_0 \sim \omega_0|_{\A_L} \otimes \omega_0|_{\A_R}$, and $\tau$ satisfies the hypotheses of Theorem \ref{theorem:stability-of-split}. Then $\omega _t = \omega_0 \circ \tau_t \sim \omega_t |_{\A_L} \otimes \omega_t |_{\A_R}$ for all $t\in [0,1]$. 
\end{corollary}

\begin{proof}
It suffices to show that $\omega_t|_{\A_L} \sim \omega_0 |_{\A_L} \circ \tau_t^L$ (resp. for the right algebra). This will follow by methods used in the proof of Theorem \ref{theorem:stability-of-split}, and so we will be brief. By a familiar asymptotic condition of being a factor state (cf. Theorem 2.6.10 of \cite{BR1}), the assumptions that $\omega_0$ is factor and $\tau$ is a quasi-local map imply $\omega_t  |_{\A_L}$ is also a factor state. Then Gronwall-type inequalities on $f_\La(t) = (\tau_t^\La - \tau_t^{\cup,\La})(A)$, $A \in \A_L \cap \A_{loc}$, show that $\omega_t |_{\A_L}$ and $\omega_0 |_{\A_L} \circ \tau_t^L$ are quasi-equivalent.
\end{proof}

\subsection{Comment on Haag duality and translation invariant states} 

Now, we consider the split property for translation invariant pure states. A result of Matsui \cite{M} shows that uniform decay of correlations in a translation invariant pure state $\varphi$ of $\A_\zz$ which satisfies Haag duality, implies $\varphi$ is split. 

\medskip

It is also proven in \cite{M} that if $\Phi$ is a translation invariant, finite-range interaction whose local Hamiltonians have a unique ground state and uniform spectral gap, $\varphi$ is a translation invariant, pure ground state of $\Phi$, and the GNS Hamiltonian $H_\varphi \geq 0$ has a nondegenerate eigenvalue at $0$, then $\varphi$ satisfies Haag duality. The conclusion is then $\varphi$ is necessarily split.

\medskip 

In the following, we remark a sufficient condition on the decay of an interaction to guarantee uniform decay of correlations, i.e. in terms of bounds which do not depend on the support size of the observables.  We leave open the question of sufficient conditions for Haag duality to hold for a general translation invariant state.  

\begin{corollary}[Uniform correlation decay]\label{cor:decay}
Suppose $\omega$ is a gapped ground state of $\Phi$, with $\norm{\Phi}_{h,\beta}<\infty$, and the GNS Hamiltonian $H_\omega \geq 0$ has a nondegenerate ground state, i.e. : 
	\begin{equation}
		\begin{split}
(i) ~\spec{H_\omega}\setminus \set{0} \subset [\gamma, \infty) ~~~\text{    and    }~~~ (ii)  ~ker(H_\omega) = \cc \Omega
		\end{split}
	\end{equation}
There exists a constant $\mu(F)>0$ such that for all $X,Y$ finite with $\max X < \min Y$, 
	\begin{equation}
		\begin{split}
| \omega(A B) -\omega( A)\omega( B) | \leq \mu(F) \norm{A}\norm{B} e^{ - uh(d(X,Y))}
		\end{split}
	\end{equation}
We may take:
	\begin{equation}
		\begin{split}
\mu(F) & = \bigg{(}  1 + \frac{ \kappa(\beta)}{\pi} + \sqrt{ \frac{ 2 \nu + \gamma}{ \pi  \gamma h(d(X,Y))}}  \bigg{)}\\
u & = \frac{ \gamma }{2 \nu + \gamma} 
		\end{split}
	\end{equation}
\end{corollary}

\begin{proof}
The proof is essentially the same as the one given in \cite{NS2} changed only to use the Lieb-Robinson bound from Corollary \ref{lem:LR}, and so we will be brief. We suppress in notation the dependence on the representation. We may assume $\ip{ \Omega, B\Omega }=0$. For free parameters $\alpha,s $, taking $b$ sufficiently small, the method of proof in \cite{NS2} gives: 
	\begin{equation}
		\begin{split}
|\omega( A \tau_{ib}(B))| = | \ip{\Omega, A \tau_{ib} (B) \Omega}| \leq \norm{A} \norm{B} \bigg{(} e^{ - \frac{\gamma^2}{4 \alpha}} + \frac{\kappa(\beta)}{\pi} e^{ \nu s -  h(d(X,Y))} + \frac{1}{s \sqrt{\pi \alpha}} e^{- \alpha s^2} \bigg{)}.
		\end{split}
	\end{equation} 
Setting $\alpha = \gamma / 2s$ and $s$ such that:
	\begin{equation}
		\begin{split}
s ( \nu + \gamma/2) =  h( d(X,Y))
		\end{split}
	\end{equation}
and taking the limit $b\to 0$ yields the bound. 
\end{proof}

\begin{proposition}
Let $\Phi$ be a translation invariant interaction on a quantum spin chain such that $\norm{\Phi}_{h,\beta} <\infty$. Suppose $\omega$ is a pure, translation invariant, gapped ground state of $\Phi$, and that the normalized GNS Hamiltonian $H_\omega$ has a nondegenerate eigenvalue at $0$. 

\medskip 

If $\omega$ satisfies Haag duality, then $\omega$ is quasi-equivalent to $\omega |_{\A_L} \otimes \omega |_{\A_R}$.
\end{proposition}

\begin{proof}
This follows immediately from Corollary 3.2 of \cite{M} and the uniform decay of correlations guaranteed by Corollary \ref{cor:decay}. 

\end{proof}

\section{Application to SPT phases}

We recall the heuristic notion of a topological phase as an equivalence class of uniformly gapped interactions, where two such interactions $\Phi_0, \Phi_1$ are related if and only if there exists a sufficiently smooth interpolating family of interactions $\Phi(s)$, $0 \leq s \leq 1$, such that $\Phi(0) = \Phi_0$ and $\Phi(1) = \Phi_1$, and $\Phi(s)$ is gapped above the ground state, uniformly in $s$. It is known that in this case, the infinite-volume ground states of $\Phi$ and $\Psi$ obtained through weak$-*$ limits of finite-volume ground states are automorphically equivalent (cf. Theorem 5.5 of \cite{BMNS}). The equivalence relation for a symmetry protected topological phase has the additional requirement that the $\Phi(s)$ are fixed by the given symmetry. The hypothesis of a uniform gap is essential, and we formulate this condition as the following working definition: Say that $\Phi$ has a \textbf{uniform gap} if there exist $\gamma>0$ and minimum interval length $R_\gamma >0$ such that if $\Lambda$ is a finite interval, $\diam{\La} \geq R_\gamma$ implies:
	\begin{equation}
		\begin{split}
 \spec{H_\La(\Phi)} = \specc_{-} (H_\La (\Phi)) \cup \specc_{+} (H_\La(\Phi))
		\end{split}
	\end{equation}
with: 
	\begin{equation}\label{eq:gap}
		\begin{split}
\min \set{ \la - \mu : \la \in\specc_{+} (H_\La(\Phi)), \mu \in  \specc_{-} (H_\La (\Phi))  } \geq \gamma
		\end{split}
	\end{equation}
and $\diam{ \specc_{-}(H_\La(\Phi)) } \to 0$ as $\diam{\La} \to \infty$. Let $\Gamma(\zz)$ denote the uniformly gapped interactions on $\zz$.

\medskip 

In the following, we also work with a familiar formulation of equivalence in a gapped SPT phase \cite{BMNS}. While we note that more general symmetries may be handled in this framework, we restrict our discussion to the antilinear $*$-automorphism $\Xi$ of time reversal since it is one of three symmetries which protect the Haldane phase in odd-spin quantum spin chains \cite{GW, O, OT, PTBO, PTBO2, T}. We do not claim that these are necessary conditions for being in the same topological phase. 

\medskip

Our application is showing that the $\zz_2$-index is a well-defined invariant for a $\Xi$-protected topological phase which contains at least one interaction with a well-defined $\zz_2$ index (e.g. a finite-range interaction), provided the decay $F_\beta$ is sufficiently strong. 

\medskip

\noindent \textbf{Assumption on decay}: Suppose $\beta>0$. Let $F_\beta$ be determined by $h(x) = R x^b$ for any $R>0$ and $b\in (0,1]$, so that (\ref{eq:Ffunction}) becomes: 
	\begin{equation}
		\begin{split}
F_\beta(x) = e^{ -R x ^b} \frac{1}{(1+x)^\beta} .
		\end{split}
	\end{equation}
We may assume, without loss of generality, that $\beta>6$. We will suppress the dependence of the $\mc{F}$-norm on the variables:
	\begin{equation}
		\begin{split}
\norm{ \cdot } _{ Rx^b, \beta} = \norm{\cdot}_F. 
		\end{split}
	\end{equation}

\begin{definition}[Equivalence in an SPT phase]\label{def:equivalence}
Define: 
	\begin{equation}\label{assumptions2}
		\begin{split}
\mf{B}(F ) = \bigg{ \{ } \Phi \in \Gamma(\zz) : ~& (i) ~ \norm{\Phi}_F < \infty, ~ (ii)~  \Phi \text{ has a unique ground state, }\\ & \hspace{13mm} (iii)~ \forall X \in P_f(\zz), ~ \Xi( \Phi(X)) = \Phi(X)      \bigg{  \} }.
		\end{split}
	\end{equation}
Define an equivalence relation $\approx$ on $\mf{B}(F)$ in the following way: $\Phi_0 \approx \Phi_1$ if there exists an interpolating path $s\mapsto \Phi(\cdot, s) \in \mf{B}(F)$ such that:

	\begin{gather}\label{assumptions}
	\begin{split}
(iv) & ~ \text{ for each }X \in P_f(\zz), ~ s  \mapsto \Phi(X,s) \text{ is continuously differentiable }\\
\\
(v) & \sup _{x,y\in \zz} \sum _{\substack{X\in P_f(\zz) \\ x,y\in X}}\sup_{s\in [0,1]}  \bigg{(} \frac{\norm{\Phi(X,s)}+ |X| \norm{\Phi'(X,s)}}{F( |x-y|)  } \bigg{)} < \infty\\
\\
(vi)& ~ \text{ the $\gamma>0$ and $R_\gamma$ in the uniform gap condition (\ref{eq:gap}) for $\Phi(\cdot, s)$ can be taken independent of $s$}.
	\end{split}
	\end{gather}
\end{definition}
Assumption (iii) of Definition \ref{def:equivalence} implies $\omega_\Phi(\Xi(A^*)) = \omega_\Phi(A)$, where $\omega_\Phi$ is the unique ground state of some representative $\Phi$. Condition (iv) of Definition \ref{def:equivalence} specifies the smoothness of the local Hamiltonians, and (v) is an assumption on the uniform spatial decay of the interactions. Precisely, (v) is sufficient decay to guarantee that the generated spectral flow will be a quasi-local map.

\subsection{Extension of the $\zz_2$ index} 

We first describe the $\zz_2$-index defined by Ogata and defer to \cite{O} for the details. Suppose $\Psi \in \mf{B}(F)$ is finite-range with pure gapped ground state $\varphi$. Since the entanglement entropy of $\varphi$ is bounded, it follows by \cite{M3} that $\varphi \sim \varphi|_{\A_L} \otimes \varphi|_{\A_R}$; and if $(\pi_R, \mf{H}_R, \Omega_R)$ is the associated cyclic representation of $\varphi |_{\A_R}$, then $\pi_R ( \A_R)''$ is a Type I factor. Hence we may assume there is an isomorphism $\iota: \pi_R( \A_R) '' \to B(\mf{K})$ for some Hilbert space $\mf{K}$. Since $\varphi|_{\A_R}$ is $\Xi$-invariant, $\Xi$ defines a unique antilinear $*$-automorphism $\hat{\Xi}$ of $B(\mf{K})$ satisfying:
	\begin{equation}
		\begin{split}
\forall A \in \A_R: ~ \hat{\Xi} \circ \iota( \pi_R(A)) = \iota \bigg{(} \pi_R \circ \Xi(A) \bigg{)}, \text{ and } \hat{\Xi}^2 = \mathrm{id}.
		\end{split}
	\end{equation}
By Wigner's theorem for antilinear $*$-automorphisms, there exists an antiunitary $J_{\pi_R} $ on $\mf{K}$, unique up to phase, such that $\hat{\Xi}(T) = J_{\pi_R} ^* T J_{\pi_R}$. Evidently $J_{\pi_R}^2 \in \set{-1,1}$, and Theorems 2.2 and 2.6 of \cite{O} show that the quantity $J_{\pi_R}^2$ does not depend on $\mf{K}$ and is an invariant of the $\approx$ relation restricted to finite-range interactions. The $\zz_2$-index is thus defined by Ogata as $\hat{\sig}_\Psi = J_{\pi_R}^2$. The extension is straightfoward to define. 

\medskip

\begin{definition}[cf. Definition 3.3 of \cite{O}]\label{def:z2-ext}
\noindent For $\Phi \in \mf{B}(F)$ with pure ground state $\omega$ such that $\omega \sim \omega|_{\A_L} \otimes \omega|_{\A_R}$, define: 
	\begin{equation}\label{eq:ext-index}
		\begin{split}
\hat{\sig}_{\Phi} = J _{\Phi}^2 \in \set{-1,1}
		\end{split}
	\end{equation}
where $J_{\Phi}$ is an antiunitary implementing the extension $\hat{\Xi}$ of time reversal to the von Neumann algebra generated by the associated cyclic representation of $\omega |_{\A_R}$.  
\end{definition}

\begin{lemma}\label{lem:SPT1}
Suppose there exists $\Phi_0 \in \mf{B}(F)$ such that the unique ground state $\omega_0$ of $\Phi_0$ is quasi-equivalent to $\omega_0 |_{\A_L} \otimes \omega_0 |_{\A_R}$. Then $\hat{\sig}_\Phi$ is well-defined for all $\Phi \in \mf{B}(F)$ such that $\Phi \approx \Phi_0$. 
\end{lemma}

\begin{proof}
Let $\Phi(\cdot, s)$, $0\leq s \leq 1$, be an interpolating path in $\mf{B}(F)$ between $\Phi_0 = \Phi(\cdot, 0)$ and $\Phi_1 = \Phi(\cdot, 1)$. By Theorem 2.2 of \cite{O}, it suffices to show that the GNS representation of the right chain restriction of $\omega$, the pure ground state of $\Phi$, generates a Type I factor. But Theorem 6.14 of \cite{NSY}, and the assumptions in (\ref{assumptions2}) and  (\ref{assumptions})  imply that the interaction $\Psi(s)$ which generates the spectral flow $\alpha_s \in \Aut (\A_\zz)$ of the family $\Phi(\cdot, s)$ satisfies the hypotheses of Theorem \ref{theorem:stability-of-split}; we may take $h(x) = O\big{(} x^b/ \log^2(x^b) \big{)}$.  

\end{proof}

\begin{proposition}\label{prop:SPT2}
If $\Phi\in \mf{B}(F)$ and  $\Phi \approx \Phi_0$, then $\hat{\sig}_{\Phi_0} = \hat{\sig}_{\Phi}$. 
\end{proposition}	

\begin{proof}
The proof direction is essentially due to Ogata in \cite{O}, and so we prove in detail only the necessary modifications to handle unbounded range of interaction. It is sufficient to show that the composition $\alpha_s \circ  [(\alpha_{s}^L)^{-1} \otimes (\alpha_{s}^R)^{-1}] $ is an inner automorphism, for all $s\in [0,1]$. Here we take the spectral flow maps to be generated by an interpolating curve $\Phi(\cdot, s)$ as in Lemma \ref{lem:SPT1}.

\medskip 

 Let $\gamma$ denote the uniform gap of the $\Phi(s)$. We show that there exists a continuous family $V(s) = V(s)^* \in \A_\zz$ such that in the uniform topology: 
	\begin{equation}\label{eq:lim}
		\begin{split}
\lim _{n \to \infty} D_{[-n,n]}(s) - D_{[-n,n]}^\cup(s) = V(s) .   
		\end{split}
	\end{equation}  
$D_{[-n,n]}(s)$ denotes the Hastings generator defined in (\ref{eq:Hastings}) of the Appendix. This implies the composition $\alpha_s \circ  [(\alpha_{s}^L)^{-1} \otimes (\alpha_{s}^R)^{-1}] $ is inner. To do this, define $g_n \in C( [0,1] , \A_\zz)$ by:
	\begin{equation}
		\begin{split}
g_n(s) = D_{[-n,n]} (s) - D_{[-n,n]}^\cup (s).
		\end{split}
	\end{equation} 
We will prove that the sequence $g_n(s)$ is uniformly Cauchy. Fix $N_0\in \nn$, and let $m,n\in \nn$ be such that $4N_0 < m \leq n$. Then: 
	\begin{equation}\label{eq:rewrite}
		\begin{split}
g_n(s) - g_m(s) & = \bigg{[} \int_{-\infty}^\infty dt~ { W_\gamma}(t)( \tau _t^{n,s} - \tau_t^{m,s}) \bigg{(} \sum _{ X \subset [-N_0, N_0]} \Phi '(X,s) \bigg{)} \\
 & \hspace{20mm} - \int _ {-\infty}^\infty W_\gamma(t)  (\tau_t^{\cup,n,s} - \tau_t^{\cup,m,s}) \bigg{(} \sum _{X \subset [-N_0, N_0]}(\Phi^\cup)'(X,s)\bigg{)} \bigg{]} \\
\\
& \hspace{25mm} + \mc{R}(n,m,N_0, s)  
		\end{split}
	\end{equation}
where $\mc{R}(n,m,N_0,s)$ is defined to be the remainder difference between $g_n(s) - g_m(s)$ and the bracketed quantity in (\ref{eq:rewrite}). Using Lemmas \ref{lem:generator-sp} and \ref{lem:generator-sp2} and the simple bound $\norm{ H_{[-N_0, N_0]}( \Phi' (\cdot, s))} \leq 3 N_0 \norm{\Phi'}_F$: 
	\begin{equation}
		\begin{split}
\norm{g_n(s) - g_m(s)} \leq 2 \big{(} 3N_0\Omega_1 (N_0) +  \Omega_2( N_0) \big{)}\norm{ \Phi'}_F 
		\end{split}
	\end{equation}
which tends to $0$ uniformly in $s$ as $N_0\to \infty$.
\end{proof}

We conclude this section with the necessary technical lemmas used in the proof of Proposition \ref{prop:SPT2}, which prove bounds analogous to those in the proof of Lemma 5.1 of \cite{O} but remain valid for interactions which are not finite-range but decay by (\ref{assumptions}). We freely use the function $I_\gamma$ defined in Lemma \ref{lem:BMNS} of the Appendix. 

\begin{lemma}\label{lem:generator-sp}
Let $\gamma>0$ and $W_\gamma$ be the weight function in (\ref{eq:Hastings}). Let $\Psi: P_f(\zz) \to \A_{loc}$ be an interaction such that $\norm{\Psi}_F<\infty$ with generated time-independent dynamics $\tau: \rr \to \Aut(\A_\zz)$. Let $\tau^n$ denote the finite-volume time-independent dynamics generated by $\Psi$ in the interval $[-n,n]$.

\medskip 

If $N, K\in \nn$ and $N\leq K$, then for all $A \in \A_{[-N,N]}$ and $n\geq m > 2K$: 

	\begin{equation}\label{eq:generator-decay}
		\begin{split}
\norm{ \int _{-\infty} ^\infty dt ~ W_\gamma(t) ( \tau_t ^n - \tau_t^m) (A) } \leq \Omega_1 \big{(}K-N \big{)}  \norm{A} 
		\end{split}
	\end{equation} 
for the decaying function:
	\begin{equation}
		\begin{split}
\Omega_1(x) & = 4 I_\gamma ( R x^b / 2\nu ) +  (\pi^2/6)^2 \bigg{(} 10 \norm{W_\gamma(t) t }_{L^1} + 2 \frac{ \kappa(\beta) \norm{W_\gamma}_{L^\infty}}{\nu} \bigg{)} \norm{ \Psi}_F e^{ - \frac{R x^b }{2\nu}}
		\end{split}
	\end{equation}
\end{lemma}

\begin{proof}
Let $T>0$ be a positive parameter. We can find a bound for the integral:
	\begin{equation}
		\begin{split}
\norm{ \int _{-T}^T dt ~ W_\gamma(t) (\tau_t^n - \tau_t^m) (A) } & \leq \int _{-T}^T dt ~ |W_\gamma(t)| \int _0^{|t|} dr ~ \norm{ [ H_{[-n,n]}(\Psi) - H_{[-m,m]}(\Psi), \tau_r ^m (A) ]}
		\end{split}
	\end{equation}
by further dividing the difference of the local Hamiltonians as: 
\begin{equation}
		\begin{split}
H_{[-n,n]}(\Psi) - H_{[-m,m]}(\Psi) =\sum _{X \in \mc{L}} \Psi(X) + \sum _{Y \in \mc{R}} \Psi(Y) + \sum _{Z \in \mc{C}} \Psi(Z)
		\end{split}
	\end{equation}
for index sets defined: 
	\begin{gather}
	\mc{L} = \set{ X \subset [-n,0] :  X \cap [-n, -m-1] \not = \emptyset  } \hspace{5mm}
	\mc{R} = \set{ Y \subset [0,n] :  ~ Y \cap  [m+1,n] \not = \emptyset  } 
	\\
	\mc{C} = \set{ Z \subset [-n,n]: Z \cap A(m,n) \not = \emptyset, ~ Z \cap (-\infty, 0] \not = \emptyset, ~ Z \cap (0, \infty) \not = \emptyset }.
	\end{gather}
First 	we bound the contribution from $\mc{L}$. For $a,b\in \zz$ such that $-n \leq a \leq -m-1$ and $a \leq b \leq 0$, denote:
	\begin{equation}
		\begin{split}
\Psi(a;b) = \sum \set{ X \in \mc{L} : \min X = a, ~ \max X = b}.
		\end{split}
	\end{equation}
Then:
	\begin{equation}\label{eq:sum1}
		\begin{split}
\int _{ - T}^T dt ~ |W_\gamma(t)| \int _0^{|t|} dr ~ \norm{ \bigg{[} \sum _{ X \in \mc{L}} \Psi(X), \tau_r^m (A) \bigg{]} } & \leq \int _{-T}^T dt ~ |W_\gamma(t)| \int _0 ^{|t|} dr ~ \sum _{ - n \leq a \leq -m-1}~ \sum _{a \leq b \leq 0} \norm{ [\Psi(a;b), \tau_r^m(A)]}.
		\end{split}
	\end{equation}	
Using Lemma \ref{lem:LR}, 
	\begin{equation}\label{eq:sum2}
		\begin{split}
\sum _{ - n \leq a \leq -m-1}~ \sum _{a \leq b \leq 0} \norm{ [\Psi(a;b), \tau_r^m(A)]} & \leq  \sum _{ - n \leq a \leq -m-1}~\bigg{(}  \sum _{a \leq b \leq a/2} \norm{ [\Psi(a;b), \tau_r^m(A)]} +  \sum _{a/2 < b \leq 0} \norm{ [\Psi(a;b), \tau_r^m(A)]}  \bigg{)} \\ 
& \leq (\pi^2/6)^2 \norm{\Psi}_F \norm{A} \bigg{(} \kappa(\beta)(e^{ \nu |r|}-1) F_{\beta-4}(K-N) + 2 F_{\beta-4}( K)    \bigg{)}.
		\end{split}
	\end{equation}	
Denote $\mc{I}_L = \int _{-T}^T dt ~ |W_\gamma(t)| \int _0 ^{|t|} dr ~ \sum _{ - n \leq a \leq -m-1}~ \sum _{a \leq b \leq 0} \norm{ [\Psi(a;b), \tau_r^m(A)]} $. Substituting (\ref{eq:sum2}) into (\ref{eq:sum1}) yields:  
		\begin{equation}\label{eq:sum3}
		\begin{split}
\mc{I}_L & \leq  (\pi^2/6)^2  \norm{ \Psi}_F \norm{A} \bigg{(} 2\norm{ W_\gamma(t) t }_{L^1} F_{\beta-4} ( K) + \frac{\kappa(\beta)\norm{W_\gamma}_{L^\infty} }{\nu} e^{ \nu T } F_{\beta-4}(K-N) \bigg{)} .
		\end{split}
	\end{equation}
By symmetry, if $\mc{I}_R$ is the corresponding integral using the interaction on $\mc{R}$, then (\ref{eq:sum3}) holds with $\mc{I}_R$ in place of $\mc{I}_L$. Next we bound the contribution from $\mc{C}$. But since these sets in $\mc{C}$ have diameter of at least $2K$, 
	\begin{equation}
		\begin{split}
\norm{ \sum _{Z \in \mc{C}} \Psi(Z)} \leq 3 (\pi^2 / 6)^2 \norm{ \Psi}_F F_{\beta-4} ( 2K).
		\end{split}
	\end{equation}
Hence we have the inequality: 
	\begin{equation}
		\begin{split}
\norm{ \int _{-\infty} ^\infty dt ~ W_\gamma(t) ( \tau_t ^n - \tau_t^m) (A) } & \leq 4 \norm{A}  I_\gamma(T) + [ 10 (\pi^2 / 6)^2  \norm{W_\gamma (t) t }_{L^1} ] \norm{\Psi}_F \norm{A} F_{\beta-4} (K) \\
& \hspace{10mm}+ \bigg{[} 2(\pi^2 /6)^2 \frac{ \kappa(\beta) \norm{W_\gamma}_{L^\infty}}{\nu} \bigg{]} \norm{\Psi}_F \norm{A} e^{ \nu T} F_{\beta-4}(K-N). 
		\end{split}
	\end{equation}
Setting $T = \frac{ R (K-N)^b}{2\nu}$ yields (\ref{eq:generator-decay}). 
\end{proof} 

It can be shown that $\lim _{n\to \infty} \int dt~ W_\gamma(t) \tau_t^n (A) = \int dt~ W_\gamma(t) \tau_t(A)$, although we do not use this fact here.  

\begin{lemma}\label{lem:generator-sp2}
Let $\gamma, W_\gamma$, and $\Psi$ be the same as in Lemma \ref{lem:generator-sp}. Suppose $K\in \nn$ and $K<n$. Then:  
	\begin{equation}\label{eq:main-ineq}
		\begin{split}
\norm{ \int _{-\infty}^\infty dt ~ W_\gamma (t) (\tau_t^n - \tau_t^{\cup, n})\bigg{(} \sum _{ \substack{ Z \subset [-n,n] \\ Z \not \subset [-K, K]}} \Psi(Z) \bigg{)} } \leq \Omega_2(K) \norm{\Psi}_F
		\end{split}
	\end{equation}
where $\tau^n, \tau^{\cup,n}$ are generated by $\Psi$ and $\Psi^\cup$, respectively, and $\Omega_2$ is the decaying function:
	\begin{equation}\label{eq:Q}
		\begin{split}
\Omega_2(x) & = 6  x \sum _{\substack{m\in \nn \\ m \geq x}}  I_{\gamma} \bigg{(} \frac{R}{2} (m/4)^b \bigg{)}+ \sum _{\substack{m\in \nn \\ m \geq x}} Q( m)
\\
Q(y)& =  (\pi^2/6)^4 \bigg{(} \frac{12 \kappa(\beta) }{\nu} \norm{\Psi}_F \norm{W_\gamma}_\infty  + 10 \max \set{\norm{\Psi}_F ^1, \norm{\Psi}_F^2 } \norm{ W_\gamma (t) |t| }_{L^1} \bigg{)} e^{  - \frac{R}{2} (y/4)^b }.
		\end{split}
	\end{equation}
\end{lemma}

\begin{proof}
First, let $J, m\in \nn$ be natural numbers such that $J < m \leq n$. Denote: 
$$\mc{I}_{m,J}= \norm{ \int _{-\infty}^\infty dt ~ W_\gamma (t) (\tau_t^n - \tau_t^{\cup, n})\bigg{(} \sum _{ \substack{ Z \subset [-m,m] \\ Z \not \subset [-J, J]}} \Psi(Z) \bigg{)}}.$$   
Furthermore, denote:
	\begin{gather}
\mc{B} = \set{X \subset [-n,n]: X \cap [-n, 0] \not = \emptyset, X \cap (0, n] \not = \emptyset} \\
\mc{D} = \set{ Z \subset [-m,m]: Z \cap \an(J+1, m) \not = \emptyset }.
	\end{gather}
Then for $T>0$, as in Lemma \ref{lem:generator-sp}, 
	\begin{equation}\label{eq:integral}
		\begin{split}
\mc{I}_{m,J}& \leq 6 (m-J) \norm{\Psi}_F I_\gamma( T)  + \int _{-T}^T dt ~ |W_\gamma(t)| \int _0^{|t|}dr ~ \norm{ \bigg{[}  \sum _{X \in \mc{B}} \Psi(X), \tau_r^n\bigg{(} \sum _{Z \in \mc{D}} \Psi(Z) \bigg{)}   \bigg{]} }. 
		\end{split}
	\end{equation}
As before, we separate the sum $\sum _{Z \in \mc{D}} \Psi(Z)$ into left, right and centrally localized terms of the interaction: 
	\begin{gather}
\sum _{Z \in \mc{D}} \Psi(Z)  = \sum _{X \in \mc{L}} \Psi(X) + \sum _{Y \in \mc{R}} \Psi(Y) + \sum _{ Z \in \mc{C}} \Psi(Z) \\
\mc{L}  = \set{X \in \mc{D}: X \subset [-m, 0]  }, \hspace{ 5mm} \mc{R}  = \set{ Y \in \mc{D} : Y \subset (0, m] } \\
\mc{C}  = \set{ Z \in \mc{D} : Z \cap [-m, 0]\not = \emptyset, ~ Z \cap (0, m] \not = \emptyset } .
	\end{gather}	
We first control the contribution to the integral from $\mc{L}$. We start this by gathering the interactions of $\mc{L}$ by intervals into $\Psi_\mc{L}(a;b) = \sum \set{ \Psi(X): X \in \mc{L}, ~ \min X = a, ~ \max X = b}$: 
	\begin{equation}\label{eq:Lcont}
		\begin{split}
\sum _{W \in \mc{L}} \Psi(W) & = \sum _{ \substack{-m \leq a \leq -J-1 \\ a \leq b \leq a/2 }} \Psi_\mc{L} (a; b)  + \sum _{ \substack{-m \leq a \leq -J-1 \\ a/2 < b \leq 0 }} \Psi_\mc{L} (a; b) := \Psi_\mc{L}^1 +  \Psi_\mc{L}^2 .
		\end{split}
	\end{equation}
Let $I _a = [ - |a/4|, |a/4| ]$. Then: 
	\begin{equation}
		\begin{split}
\norm{ \bigg{[} \tau_r^n (  \Psi_\mc{L}^1  ), \sum _{ X\in \mc{B}} \Psi(X)  \bigg{]} } & \leq \sum _{ -m \leq a \leq -J-1}\bigg{(}  \norm{ \bigg{[}\sum _{a \leq b \leq a/2}  \tau_r^n(\Psi_{\mc{L}}(a; b)) , \sum _{X \in \mc{B}: X \subset I_a} \Psi(X) \bigg{]} }\\
& \hspace{35mm} +  \norm{ \bigg{[}\sum _{a \leq b \leq a/2}  \tau_r^n(\Psi_{\mc{L}}(a; b)) , \sum _{X \in \mc{B}: X \not \subset I_a} \Psi(X) \bigg{]} } \bigg{)}. \\ 
		\end{split}
	\end{equation}
By applying Lieb-Robinson bounds, the following inequality is valid:
	\begin{equation}
		\begin{split}
\norm{ \bigg{[}\sum _{a \leq b \leq a/2}  \tau_r^n(\Psi_{\mc{L}}(a; b) ), \sum _{X \in \mc{B}: X \subset I_a} \Psi(X) \bigg{]} } &\leq \sum _{a \leq b \leq a/2} \kappa(\beta) \norm{\Psi}_F^2 F_{\beta}(|b-a|) |a| (e^{\nu |r|}-1) F_{\beta-2}(|a/4|)\\
 & \leq  \frac{\pi^2}{6} \kappa(\beta) \norm{\Psi}_F^2 (e^{\nu |r|} - 1) \frac{ |a|}{(1+ |a/4|)^{3}} F_{\beta-5} ( |J/4|) 
		\end{split}
	\end{equation}
and the right-hand side is summable in $|a|$. And by both decay of the interaction and application of Lieb-Robinson bounds, 
	\begin{equation}
		\begin{split}
 \norm{ \bigg{[}\sum _{a \leq b \leq a/2}  \tau_r^n(\Psi_{\mc{L}}(a; b) ), \sum _{X \in \mc{B}: X \not \subset I_a} \Psi(X) \bigg{]} }  & \leq   \norm{ \bigg{[}\sum _{a \leq b \leq a/2}  \tau_r^n(\Psi_{\mc{L}}(a; b)) , \sum _{\substack{-n \leq c < a/4 \\ 0 \leq d \leq n }} \Psi_\mc{B}(c;d) \bigg{]}} \\
 & \hspace{25mm} +  \norm{ \bigg{[}\sum _{a \leq b \leq a/2}  \tau_r^n(\Psi_{\mc{L}}(a; b)) , \sum _{\substack{ a/4 \leq c \leq 0 \\ |a/4| < d \leq n}} \Psi_\mc{B}(c;d)  \bigg{]} }\\
 & \leq 2\kappa(\beta) (\pi^2 / 6)^3  \norm{\Psi}_F^2  \frac{1}{(1+ |a/4|)^2}\bigg{(}  e^{ \nu |r| } F_{\beta-2} ( J/4) \bigg{)}
		\end{split}
	\end{equation}
where $\Psi_{\mc{B}}(c;d)$ is defined as $\Psi_{\mc{L}}(a;b)$ only with respect to the index set $\mc{B}$. Hence:
	\begin{equation}
		\begin{split}
\norm{ \bigg{[} \tau_r^n (  \Psi_\mc{L}^1  ), \sum _{ X\in \mc{B}} \Psi(X)  \bigg{]} }  \leq 6\kappa(\beta) (\pi^2/6)^4 \norm{\Psi}_F^2  e^{ \nu |r|} F_{\beta-5} ( J/4).
		\end{split}
	\end{equation}
And again by decay of the interaction:
	\begin{equation}\label{eq:Psi2}
		\begin{split}
\norm{ \bigg{[} \tau_r^n (  \Psi_\mc{L}^2  ), \sum _{ X\in \mc{B}} \Psi(X)  \bigg{]} } & \leq 2\norm{\tau_r^n (  \Psi_\mc{L}^2  )}\norm{ \sum _{ X\in \mc{B}} \Psi(X)   } \leq2 ( \pi^2 / 6)^4 \norm{\Psi}_F^3 F_{\beta-4}(J/2)  .
		\end{split}
	\end{equation}	
By symmetry on the chain about $0$, this majorizes the contribution from $\mc{R}$ as well. And decay of the interaction also yields a bound on the contribution from $\mc{C}$ in the same manner as in (\ref{eq:Psi2}):
	\begin{equation}
		\begin{split}
\norm{ \bigg{ [ } \sum _{X \in \mc{B}} \Psi(X) , \tau_r^n \bigg{(}\sum _{Z \in \mc{C}} \Psi(Z) \bigg{)} \bigg{]} }\leq 3 (\pi^2/6)^4 \norm{\Psi}_F ^2 F_{\beta-4} (J).
		\end{split}
	\end{equation}
Hence if we set $T = \frac{R}{2\nu} (J/4)^b$, the integral expression $$\mc{J} = \int _{-T}^T dt ~ |W_\gamma(t)| \int _0^{|t|} dr ~ \norm{ \bigg{[}  \sum _{X \in \mc{B}} \Psi(X), \tau_r^n\bigg{(} \sum _{Z \in \mc{D}} \Psi(Z) \bigg{)}   \bigg{]} }$$ of the right-hand side of the inequality (\ref{eq:integral}) is bounded: 
	\begin{equation} \label{eq:J-ineq}
		\begin{split}
\mc{J} & \leq \frac{12 \kappa(\beta) (\pi^2/6)^4}{\nu}  \norm{\Psi}_F^2 \norm{W_\gamma}_\infty e^{  - \frac{R}{2} (J/4)^b } \\
& \hspace{35mm}+ 10(\pi^2/6)^4 \max \set{\norm{\Psi}_F ^2, \norm{\Psi}_F^3 } \norm{ W_\gamma (t) |t| }_{L^1} F_{\beta-4}(J/2). \\
		\end{split}
	\end{equation}
The right-hand side of the inequality (\ref{eq:J-ineq}) is bounded above by $Q(J)$, as defined in (\ref{eq:Q}). Then (\ref{eq:integral}) continues as: 
	\begin{equation}
		\begin{split}
\mc{I}_{m,J} \leq 6 (m-J) \norm{\Psi}_F I_\gamma \bigg{(}  \frac{R}{2} (J/4)^b \bigg{)} + Q(J).
		\end{split}
	\end{equation}
Now we prove the inequality (\ref{eq:main-ineq}). There exists a maximal $M_0 \in \nn$ such that $n > M_0 K$, and so:
	\begin{equation}
		\begin{split}
\mc{I}_{n,K} & = \mc{I}_{2K, K} + \mc{I}_{3K, 2K} + \ldots + \mc{I}_{M_0 K , (M_0-1)K} + \mc{I}_{n, M_0 K} \\
& \leq 6 \norm{ \Psi}_F K \sum _{j=1}^{ M_0}  I_{\gamma} \bigg{(} \frac{R}{2} (jK/4)^b \bigg{)}+ \sum_{j=1}^{M_0}Q(jK) .
		\end{split}
	\end{equation}
Set $H_\gamma(x) = 6 \norm{ \Psi}_F x \sum _{\substack{m\in \nn \\ m \geq x}}  I_{\gamma} \bigg{(} \frac{R}{2} (m/4)^b \bigg{)}+ \sum _{\substack{m\in \nn \\ m \geq x}} Q( m) .$
\end{proof}

\medskip 

\subsection{Acknowledgments} The author thanks B. Nachtergaele for many helpful comments and discussions on the subject of this manuscript. This work was supported by the NSF grant DMS 1813149 and a Simons-CRM research grant during the 2018 Thematic Semester, \textit{Mathematical challenges in many-body physics and quantum information}, at the Centre de recherches math\'ematiques.

\medskip



\section{Appendix: Generator of the spectral flow}

In this appendix we briefly recall notations and properties of the spectral flow. For a more detailed analysis of quasi-locality and symmetries of the spectral flow, see e.g. Sections 6 and 7 of \cite{NSY} and Proposition 5.4 of \cite{BMNS}. In finite volume $\La$, the spectral flow is implemented for gapped, continuously differentiable families of Hamiltonians $H_\La(s)$ by unitiaries solving:
	\begin{equation}
		\begin{split}
\frac{d}{ds} U_\La(s) = i D_\La(s) U_\La(s), ~ U_\La(0) = \1
		\end{split}
	\end{equation}
for the Hastings generator:
	\begin{equation}\label{eq:Hastings}
		\begin{split}
D_\La(s) = \int _{-\infty}^\infty dt~ W_\gamma(t) \tau_t ^{\La,s} \bigg{(} \frac{d}{ds} H_\La(s) \bigg{)} .
		\end{split}
	\end{equation}
Here $\gamma>0$ refers to the uniform gap of the $H_\La(s)$, and $W_\gamma \in L^1 \cap L^\infty$ is chosen as the odd function, positive on $(0,\infty)$ from Equation (2.12) of \cite{BMNS}. Explicit estimates on the integral $I_\gamma(t) = \int _t ^\infty dr~ W_\gamma(r) \geq 0$ are known:

\begin{lemma}[Lemma 2.6 of \cite{BMNS}]\label{lem:BMNS}
For $t> 36058$, 
	\begin{equation}
		\begin{split}
I_\gamma(t) \leq [130 e^2 \gamma^9] t^{10} \exp\bigg{(} -\frac{2}{7} \frac{\gamma t}{ (\ln(\gamma t))^2} \bigg{)} .
		\end{split}
	\end{equation}
\end{lemma}

\end{document}